%% file: Conductance-Arxiv2024.tex
\pdfoutput=1
\documentclass[11pt]{article}
\usepackage[margin = 1 in]{geometry}
\usepackage{mathtools}
\usepackage{amsthm,amsmath}
\usepackage{amssymb}
\usepackage{stmaryrd}
\usepackage[noend]{algpseudocode}
\usepackage{caption}
\usepackage{relsize}
\usepackage{lineno}
\usepackage{graphicx}
\newtheorem{theorem}{Theorem}
\newtheorem{lemma}{Lemma}

\newtheorem{definition}{Definition}

\newtheorem{claim}{Claim}

\newtheorem{corollary}{Corollary}

\newtheorem{question}{Question}
\newtheorem{observation}{Observation}
\newtheorem*{theorem*}{Theorem}

\usepackage{hyperref}

\title{A Distributed Conductance Tester Without \\ Global Information Collection}
\author{Tu\u{g}kan Batu \\London School of Economics
  and Political Science, UK \\ \texttt{t.batu@lse.ac.uk} \and Amitabh Trehan \\Durham University, UK \\ \texttt{amitabh.trehan@durham.ac.uk}
  \and Chhaya Trehan \\London School of Economics and Political Science and University of Bristol, UK\\ \texttt{chhaya.dhingra@gmail.com} }

\usepackage[utf8]{inputenc}
\usepackage[T1]{fontenc}

\usepackage{fullpage}
\usepackage{times}
\usepackage{lmodern}
\usepackage{booktabs}
\usepackage{threeparttable}
\usepackage{amsmath}
\usepackage{amssymb}

\usepackage{thmtools}
\usepackage{thm-restate}
\usepackage{hyperref}
\usepackage{thm-restate}
\usepackage{cleveref}

\usepackage[ruled]{algorithm}

\usepackage{bbm} 

\DeclareMathOperator{\volume}{vol}         
\DeclareMathOperator{\trap}{trap}         
\DeclareMathOperator{\degree}{deg}        
\newcommand{\R}{\mathbb{R}}               
\newcommand{\comp}[1]{\overline{#1}}      
\newcommand{\trans}[1]{{#1}^{\mathsf{T}}} 
\newcommand{\ind}[1]{{\mathbbm{1}_{#1}}}  
\newcommand{\trind}[1]{\mathbbm{1}_{#1}^{\mathsf{T}}}
\newcommand{\eqdef}{\coloneqq} 
\algnewcommand{\LeftComment}[1]{\Statex \(\triangleright\) #1}

\usepackage{mathtools}
\DeclarePairedDelimiter{\norm}{\lVert}{\rVert}
\newcommand{\thresholdCount}{\tau}

\begin{document}
%
%
%


%
\maketitle              
\begin{abstract}
We propose a simple and time-optimal algorithm for property testing a graph for its conductance in the CONGEST model. Our algorithm takes only $O(\log n)$ rounds of communication (which is known to be optimal), and consists of simply running  multiple random walks  of  $O(\log n)$ length from a certain number of random sources, at the end of which nodes can decide if the underlying network is a  good conductor or far from it. Unlike previous algorithms, no aggregation is required even with a smaller number of walks. Our main technical contribution involves a tight analysis of this process for which we use spectral graph theory. 
We introduce and leverage the concept of sticky vertices which are vertices in a graph with low conductance such that short random walks 
originating from these vertices end in a region around them.

The present state-of-the-art distributed CONGEST algorithm for the problem by Fichtenberger and Vasudev [MFCS 2018], 
  runs in $O(\log n)$ rounds  using three distinct phases : building a rooted spanning tree (\emph{preprocessing}), running $O(n^{100})$ random walks to generate statistics (\emph{Phase~1}), and then convergecasting to the root to make the decision (\emph{Phase~2}).  The whole of our algorithm is, however, similar to their Phase~1 running only  $O(m^2) = O(n^4)$ walks.
  Note that aggregation (using spanning trees) is a popular technique but spanning tree(s) are sensitive to node/edge/root failures, hence, we hope our work points to other more distributed, efficient and robust solutions for suitable problems.
  
\end{abstract}
\input{intro-arxiv2024.tex}
\input{prelim-arxiv2024.tex}

\input{algorithm-arxiv2024.tex}
\input{conclusions-arxiv2024.tex}

%
%
\bibliographystyle{alphaurl}
\bibliography{biblioExpansion}
\end{document}

%% file: intro-arxiv2024.tex
\section{Introduction}
\label{sec:Intro}
Checking whether a distributed network satisfies a certain property is an important problem. For example, this knowledge may be used to choose appropriate algorithms to be run on the network for certain tasks.
For instance, the randomised leader election algorithm of~\cite{leaderElection} works in sublinear time if the underlying graph is a good expander but not otherwise.
However, it may be hard to efficiently \emph{verify} certain
global graph properties in the CONGEST model of distributed computing.  In
this model, each vertex of the input graph acts as a processing unit and works in conjunction with other vertices to solve a computational problem. The
computation proceeds in synchronous rounds, in each of which every vertex can
send an~$O(\log n)$-bits message to each of its neighbours, do some local
computations and receive messages from its
neighbours. 

\emph{Distributed decision} problems are tasks in which the vertices of the underlying network have to collectively decide whether the network satisfies a global property $\mathcal P$ or not.  If the network indeed satisfies the property, then all vertices must accept and, if not, then at least one vertex in the network must reject. For many global properties, lower bounds on the number of rounds of computation of the form $\tilde \Omega(\sqrt n + D)$ are
known for distributed decision, where $n$ is the number of
vertices and $D$ is the diameter of the network.
(See~\cite{hardnessOfVerification}). It makes sense to relax
the decision question and settle for an approximate answer in these scenarios as is done in the field of property testing (see~\cite{propertyTesting1,propertyTesting2}) in the sequential setting.
A property testing algorithm in the sequential setting arrives at an
approximate decision about a certain property of the input by \emph{querying}
only a small portion of it. Specifically, an $\epsilon$-tester for a graph
property $\mathcal P$ is a randomised algorithm that can distinguish between
graphs that satisfy $\mathcal P$ and the graphs that are $\epsilon$-far from
satisfying $\mathcal P$ with high constant probability. An $m$-edge graph~$G$
is considered $\epsilon$-far from satisfying $\mathcal P$ if one has to modify
(add or delete) more than $\epsilon \cdot m$ edges of~$G$ for it to satisfy
$\mathcal P$. Two-sided error testers may err on all graphs, while one-sided
error testers have to present a witness when rejecting a graph. The cost of
the algorithm is measured in the number of queries made.
(See~\cite{propertyTesting1,propertyTesting2,propertyTesting3,graphPropertyTesting} for a detailed exposition of the subject.)

\noindent\textbf{Distributed Property Testing:}
A \emph{distributed property testing} problem is a relaxed variant of the
corresponding decision problem: if the input network
satisfies a property, then, with sufficiently high probability, all the vertices
accept  but if the input network is $\epsilon$-far from satisfying the
property, then at least one vertex rejects. The definition of "farness" in it 
remains the same as in the classical setting.  The
complexity measures are the number of communication rounds and
the number of messages exchanged during the execution of the tester.
Distributed property testing has been an active area of research recently. The
work of~\cite{brakerskiPatt-Shamir} was the first to present a distributed
algorithm (for finding near-cliques) with a property testing flavour. Later, in~\cite{Censor-Hillel2019}, the authors  did a
more detailed study of distributed property testing.  There has been further
study on the topic (see~\cite{even_et_al2017} and~\cite{FischerGO17}) in the
specific context of subgraph freeness. 

\noindent\textbf{Conductance Testing:}
\label{sec:exptesting}
We address the problem of testing the conductance of an unweighted, undirected
graph $G = (V,E)$ in the CONGEST model. Throughout, we
denote~$|V|$ by~$n$ and $|E|$ by $m$.  A distributed conductance tester can be
a useful pre-processing step for some distributed algorithms (such
as~\cite{leaderElection}) which perform better on graphs with high conductance. 
The test can help determine whether to proceed with the algorithm or not.

Given a graph $G = (V,E)$, and a set $A \subseteq V$, the \emph{volume} of $A$
(denoted~$\volume(A)$) is the sum of degrees of vertices in $A$.
We say that a graph $G =(V,E)$ is an \emph{$\alpha$--conductor} if every
$U\subseteq V$ such that~$\volume(U) \le \volume(V)/2$ has conductance at
least~$\alpha$.  Here, the conductance of a set $A$ is defined as
$E(A, V\setminus A)/\volume(A)$, where $E(A, V\setminus A)$ is the
number of edges crossing between~$A$ and~$V\setminus A$.  A closely related
property of graphs is their expansion.  We call a graph $G =(V,E)$ an
\emph{$\alpha$--vertex expander} if every $U\subseteq V$ such
that~$|U| \le |V|/2$ has at least~$\alpha |U|$ neighbours. Here, a vertex
$v \in V\setminus U$ is a neighbour of $U$ if it has at least one edge incident
to some~$u \in U$. Similarly, $G$ is called an \emph{$\alpha$--edge expander},
if every $U\subseteq V$ such that~$|U| \le |V|/2$ has at least~$\alpha |U|$
edges crossing between $U$ and $V \setminus U$.
For a constant $d$, a graph $G = (V,E)$
is called a \emph{bounded-degree graph} with degree bound~$d$ if every
$v \in V$ has degree at most $d$. In this case, both the vertex and edge expansions are bounded by
a constant (depending on $d$) times  the conductance. Testing expansion (essentially
testing conductance) in the bounded degree model has been studied for a long time in the 
classic centralised property testing model.
In this setting, the problem of testing expansion was
first studied by Goldreich and Ron~\cite{GoldreichRonTester} and later by~\cite{CzumajS10}.
Specifically, Czumaj and Sohler showed that given parameters $\alpha, \epsilon > 0$, the tester 
proposed by~\cite{GoldreichRonTester}
accepts all graphs with vertex
expansion larger than $\alpha$, and rejects all graphs that are $\epsilon$-far from having vertex
expansion at least $\alpha'  = \Theta( \alpha^2/\log n )$.
Their work was followed by the state of the art results by~\cite{KaleS11}
and~\cite{AsafAsaf10}.  Both these papers present
$\tilde O(n^{1/2 + \mu} / \alpha^{2})$-query testers (for a small constant $\mu > 0$) for distinguishing between
graphs that have expansion at least $\alpha$ and graphs that are $\epsilon$-far
from having expansion at least $\Omega(\alpha^2)$. 
In the general graph model (with no bound on degrees),
Li, Pan and Peng~\cite{Li2011TestingCI} presented a 
conductance tester. Their tester essentially pre-processes
the graph and turns it into a bounded degree graph 
while preserving (roughly) its expansion and size and then 
uses a tester for bounded degree graphs.

In the last few years, the
same problem has also been addressed in non-sequential models of computing such
as MPC~\cite{MPCtester} and distributed CONGEST~\cite{VasudevDistributed}. 
There are earlier papers studying distributed random walks whose results can be adapted towards conductance testing e.g.~\cite{MP-SparseCut,MP-Mixing-ICDCN17} . 
However, these results yield large gaps in the conductance of the graphs that are accepted ($\Omega(\alpha)$) and 
that of the graphs that are rejected ($O(\alpha^2/\mathrm{polylog}(n)) $).
The first distributed algorithm for the specific task of testing the conductance of an input graph that we are aware of 
 is by Fichtenberger and
Vasudev~\cite{VasudevDistributed}. This can test the conductance of the input network in
the unbounded-degree graph model (like ours). 

A typical algorithm for the problem in the sequential, as well as
non-sequential, models can be thought of as running in two \emph{phases} (after possibly a \emph{pre-processing phase}).  In
the first phase, the algorithm performs a certain number of short
($O(\log n)$-length) random walks from a randomly chosen starting vertex.  The
walks should mix well on a graph with high conductance and should take longer to
mix on a graph which is far from having high conductance (at least from some
fraction of starting vertices).  In the second phase, the algorithm then checks
whether those walks mixed well or not.  For that, the algorithm gleans some
information from every vertex in the graph and computes some aggregate
function.  Specifically in the classical and MPC settings, the algorithms count
the total number of pairwise collisions between the endpoints of the walks. 

The distributed algorithm for the problem
by Fichtenberger and Vasudev~\cite{VasudevDistributed} precedes the first phase by building a rooted BFS
spanning tree of the input graph.\footnote{If the construction of BFS tree
  takes longer than $O(\log n)$ rounds the algorithm rejects without proceeding
  to the first phase since all good conductors have small diameter. However, a
  bad conductor such as a dumbbell graph may also have small diameter, so their
  algorithm still needs to proceed with the test after the successful
  construction of the spanning tree.}  This spanning tree is used for
collecting information from the endpoints of the random walks in the second
phase.  Specifically, their algorithm estimates the discrepancy of the endpoint
probability distribution from the stationary distribution by collecting the estimate of
discrepancy on each endpoint at a central point.  If the overall discrepancy is above a
certain threshold, the algorithm rejects the graph. This process of building a
spanning tree and collecting information at the root to decide if the property
holds or not takes a global and centralised view of the testing process.

The following natural question arises in the context of the second phase:
\vspace*{-0.5em}
\begin{question}\label{qust: noGlobalInfo}
Is it possible to execute the second phase without computing a global aggregate function?
\end{question}
\vspace*{-0.5em}
In the classic setting, one strives for testers that make a sublinear (in $n$ or $m$)
number of queries which translates to running a sublinear number of walks.
With only a sublinear $(O (\sqrt n))$ number of walks, one hardly expects to
see any useful information by itself on any individual vertex or in a small
constant neighbourhood around it to know if the walks mixed well or
not. Therefore, one has to rely on an aggregate function such as the total
number of pairwise collisions between the endpoints of the walks.
In the non-sequential settings such as distributed CONGEST, one can utilise
the parallelism to run a superlinear number of short walks while keeping the run
time proportional to the length of the walks. This inspires us to stick to
Question~\ref{qust: noGlobalInfo} in distributed setting and investigate what
information one should store at each vertex during phase $1$ and how it should
be processed \emph{locally} to allow each node to decide locally
whether it is part of a good conductor or not. We investigate and answer the 
 following question in affirmative for general graphs (with no degree bound), 
 \vspace*{-0.5em}
 \begin{question}\label{qust: noGlobalInfoDist}
   Can we test the conductance in distributed CONGEST model without relying on any
   global information collection at a central point?
   \end{question}
   \vspace*{-0.5em} 
This leads us to our main result 
provided all the nodes know $n$ and $m$ beforehand.
Note that one can overcome the requirement of knowing $m$ by performing a rooted
spanning tree construction as in~\cite{VasudevDistributed} and using this tree 
to count the number of edges. Note that we will not use this tree for
collecting information about the random walks. \\
  \noindent \textbf{Our Results:} 
  Our main result is the algorithm presented in Section~\ref{sec:alg} (Pseudocodes~\ref{alg:disttester} and~\ref{alg:moveOneStep}). The main theorem is restated below:
 \begin{theorem*} [\textbf{Theorem~\ref{thm:mainthm}}]
  For an input graph $G = (V,E)$, and
  parameters~$0 < \alpha < 1$ and~$\epsilon > 0$, 
the distributed algorithm described in Section~\ref{sec:alg} 
\begin{itemize}
\item outputs Accept, with probability at least $2/3$, on
  every vertex of $G$ if $G$ is an $\alpha$-conductor.
\item outputs Reject, with probability at least $2/3$, on at least one
  vertex of $G$ if $G$ is $\epsilon$-far from any~$(\alpha^2/2880)$-conductor. 
\end{itemize}
The algorithm uses $O_{\epsilon, \alpha}(\log n)$ communication rounds.
\end{theorem*}
To be precise, the algorithm runs $2m^2$ walks of length $\frac{32}{\alpha^2} \log n$  from  each of $\theta(1/\epsilon)$ starting vertices with the number of communication rounds equal to  length of each random walk i.e. $\frac{32}{\alpha^2} \log n$ and messages at most $O(m \log n)$. A lower bound theorem in~\cite{VasudevDistributed} (Theorem 2) states that any distributed tester with this gap requires $\Omega(\log(n + m))$ rounds of communication even in LOCAL model of distributed computing. 
This implies that our running time in the more restricted CONGEST is optimal. 
 
\noindent\textbf{Testing in a single phase:}
\label{head:messageSaving}
The advantage of not having to collect global information is that it lets us do away with the
wasteful construction of a spanning tree and information accumulation at the
root.  Since we do not need to construct a spanning tree, we do not need a pre-processing phase unlike~\cite{VasudevDistributed}.
Note that setting up a spanning tree creates multiple points of failures for the aggregation phase. One could attempt to handle failure of the root of a single tree by setting up multiple spanning trees simultaneously. However, note that a single node failure (of a node internal to all these trees) could disconnect all these trees and if this happens early in the phase $2$, we may not get enough information for the root(s) to make their decisions. Since our phase $2$ is `instant' i.e. involves no communication, we do not have any failure issues. This opens up the possibility of a
fully-fault-tolerant tester for dynamic networks if a fault-tolerant phase $1$ (i.e. fault-tolerant random walks) could be designed.

 \subsection{Technical Overview}
 In this section, we give a general overview of the concepts
 used in our algorithm.  Like all the previous algorithms
 for conductance testing, we perform a certain number of random walks from a randomly selected
 starting vertex.  To boost the success probability of the process, we repeat
 this process in parallel from a constant number of randomly selected starting
 vertices.  The main technical challenge in running random walks in parallel
 from different starting vertices is the congestion on the edges.  As done
 by~\cite{VasudevDistributed}, we overcome this problem by not sending the
 entire trace of the walk from its current endpoint to the next.  For each
 starting point~$q$ and for all the walks going from $u$ to $v$, we simply send
 the ID of $q$ and the number of walks destined for $v$ to $v$.  At the end of
 this process, for each starting point $q$, we simply store at each vertex $v$,
 the number of walks that ended at $v$. Finally, each vertex $v \in V$
 looks at the information stored at $v$ to check if the
 number of walks received from any starting vertex is more than a certain threshold.
 If so, it outputs \emph{Reject} and, otherwise, it outputs \emph{Accept}.
 
  \noindent\textbf{Stickyness Helps:} To show that the number of walks received by a vertex $v$ is sufficient to decide
 whether~$v$ is part of a good conductor or not, we proceed as follows.  A technical lemma from~\cite{LiP15} implies that if a graph $G$ is
 $\epsilon$-far from being an $\alpha$-conductor, then there exists a set
 $S \in V$ of sufficiently low conductance (of cut $(S, V\setminus S)$, see
 Definition~\ref{def:cutConductance})) and sufficiently high volume. It follows intuitively that it is
 likely that a short random walk starting from a randomly selected starting
 vertex in $S$ should not go very far and end in $S$. In particular, we show
 that there exist a subset $P \subseteq S$ such that short walks starting
 from any $v \in P$ end in a large enough region $T$ (subset of $S$) around $v$.  We make
 this notion precise by using spectral graph theory to show that a large portion of the volume of
 low-conductance set~$S$ (as described above) belongs to  \emph{sticky}
 vertices.  We call a vertex $v \in S$ sticky if there exists a set
 $T \subseteq S$ such that $v \in T$ and short random walks starting from $v$
 end in~$T$ with a \emph{sufficiently} high probability.  We define
 $\trap(v, T, \ell)$ as the probability that an $\ell$ length walk starting
 from~$v \in T \subseteq S$ ends in $T$. 
 
 \noindent\textbf{Trap Probability:} Observe that our definition of trap probability is slightly different from the one
 generally used in the analysis of similar problems.
 The notion of trap probability is generally used to bound the probability of an $\ell$-step 
 random walk staying in a specific set for its entire duration (in each of the $\ell$ steps). See for
 example the definitions of \emph{remain} and \emph{escape} probabilities in~\cite{GharanT12}.
 Similarly, Czumaj and Sohler also implicitly use the concept of trap probability in their expansion tester~\cite{CzumajS10} 
 and they also bound the probability of a walk staying inside a set of low conductance
 for its entire duration. We relax the definition a bit and only care about the walk \emph{ending} in a subset of a low  conductance set. 
 This allows us to also use the walks that may have briefly escaped the low conductance 
 region when counting the number of trapped walks.
 Thus, if we run sufficiently many walks from one of the sticky vertices, then a lot
 of them will end in a subset $T$ of $S$ and some vertex in $T$ will see a lot
 more walks than any vertex in a good conductor should.  To ensure that we pick
 one of the sticky vertices as a starting vertex, we sample each vertex to be a
 starting vertex with appropriate probability.

\noindent\textbf{Spectral Approach:} In the analysis of the convergence behaviour of random walks (to the stationary distribution) using the eigenvectors and eigenvalues of the random walk matrix $M$ (first introduced by Kale And Seshadhri~\cite{KaleS11} and later refined by~\cite{VasudevDistributed} for unbounded degree graphs), they divide the set of eigenvectors of $M$ into \emph{heavy} and \emph{light} sets. 
All the eigenvectors with eigenvalues above a certain threshold (appropriately chosen) are considered heavy and the rest and considered light. 
This lets one drop all the light eigenvectors from the analysis since their contribution to the convergence behaviour is minimal. 
In our analysis, we use a similar technique where we focus on the heavy eigenvectors of the random walk matrix $M$
to lower bound the trap probability of a random walk from an appropriately chosen starting vertex. 
To further tighten our analysis, we further divide the set~$H$ of heavy eigenvectors into the \emph{heaviest} eigenvector $\vec e_1$ (with the maximum eigenvalue $1$) and the set $H \setminus \{\vec e_1\}$. 
We use both (but separately and not as one bundle $H$) in our analysis.
This also makes intuitive sense since the heaviest eigenvector
makes the maximum contribution to the trap probability and treating it separately tightens our bound.
We note that~\cite{KaleS11} and~\cite{VasudevDistributed}
analyse a different measure - the discrepancy between their final endpoint probability distribution and the stationary distribution;
and the contribution of $\vec e_1$ to this measure is zero, so their analysis does not benefit from segregating the heaviest eigenvector from the set $H$ of the heavy eigenvectors. \\
\noindent\textbf{Organisation:}
The rest of the paper is organised as follows.  In Section~\ref{sec:Prelim}, we
provide necessary definitions and state some basic lemmas that are used in rest
of the paper. In Section~\ref{sec:alg}, we provide a detailed description of our testing algorithm.  Section~\ref{sec:combproof} is
dedicated to the proof of our main theorem.

%% file: prelim-arxiv2024.tex
\section{Preliminaries}\label{sec:Prelim}
Let $G = (V,E)$ be an unweighted, undirected graph on $n$ vertices and $m$ edges.  
We assume that the vertices of $G$ have unique identifiers.
For a given vertex $v \in V$, $\degree(v)$ denotes the \emph{degree} of $v$.  
For sets~$A,B\subseteq V$, we denote by~$E(A,B)$ the number of edges that have one
endpoint in~$A$ and the other 
in~$B$.
A cut 
is a partition of the vertices
into two disjoint subsets.
Given a graph~$G = (V,E)$ any subset $S \subseteq V$ defines a cut denoted by $(S, \comp{S})$,
where $\comp{S} = V\setminus S$.
\begin{definition}\label{def:cutConductance}
Given a cut $(S, \comp S)$ in $G$, 
the conductance of $(S, \comp S)$ is defined as 
\begin{equation*}
     \frac{E(S,\comp S)}{\min \{\volume(S), \volume(\comp S) \}},
\end{equation*}
where $\volume(A) = \sum_{v \in A} \degree(v)$. Alternatively, we also refer to
the conductance of a cut~$(S,\comp{S})$ as the \emph{conductance of
  set~$S$}. The conductance of a graph is the minimum conductance of any cut in
the graph.
\end{definition}
Throughout the paper, all vectors $\vec x \in \R^n$ are column vectors.  For a
vector~$\vec x \in \R^n$, we denote by~$\trans{\vec x}$ the transpose of
$\vec x$.  For two vectors $\vec x$ and $\vec y$ in $\R^n$,
$\langle \vec x, \vec y \rangle$ denotes their inner product.
We denote the $n \times n$ adjacency matrix of the input graph $G$ by $A$,
where $A_{ij} = 1$, if~$(i, j) \in E$ and~$0$ otherwise. Let $D$ denote the
$n \times n$ diagonal degree matrix of $G$, where~$D_{ij} = \degree(i)$ if
$i =j$ and~$0$ otherwise.
The main technical tool in our analysis will be random walks on the input graph
$G$.  We denote a random walk by its transition matrix $M$. For a pair
of vertices~$u, v \in V$,  let~$M_{uv}$ be the probability of
visiting $u$ from $v$ in one step of $M$. In the standard
definition of a random walk, $M_{uv}$ is defined as $1/\degree(v)$ if there is an edge from~$u$ to~$v$, and $0$ otherwise.

We use a slightly modified version of the standard random walk called a
\emph{lazy random walk}.  
A \emph{lazy random walk} currently
stationed at $v \in V$, stays at $v$ with probability~$1/2$ and, with the
remaining probability $1/2$, it visits a neighbour of $v$ uniformly at random.
 Let $M$ be a lazy random walk on $G$,
the transition probabilities for $M$ are defined as follows: for a
pair~$v,w \in V$ such that $v \neq w$, 
$M_{wv} = \frac{1}{2\degree(v)}$, if
$(v,w) \in E$ and~$0$, otherwise. Further, for $v \in V$, we 
define~$M_{vv} = 1/2$.
Algebraically, $M$ can be expressed as
$M = \frac{1}{2} (I + AD^{-1}),$ where $I$ is the $n \times n$ identity matrix.
Let $\pi$ be the stationary distribution of $M$.
In the stationary distribution of a lazy random walk,
each vertex is visited with probability proportional to its degree.
More formally,
$
\pi = \frac{\vec d}{2m},
$
where $\vec d$ is an $n$-dimensional vector of vertex degrees and $m$ is the number of edges in $G$.

In the following, we provide a brief exposition of relevant concepts from spectral graph theory.
We refer the reader to the textbook~\cite{chung97} by Fan Chung for a detailed treatment of the subject.
Note that for irregular graphs, $M$ is an asymmetric matrix and may not have an orthogonal set of eigenvectors.
For analyzing random walks on $G$ in terms of the eigenvalues and eigenvectors of its associated matrices, 
we rely on a related symmetric matrix called the \emph{normalized Laplacian} of $G$ denoted by $N$.
The normalized Laplacian $N$ is defined as
\[
N = I - D^{-1/2} A D^{-1/2}.
\]
We show below a way to express $M$ in terms of $N$.
\begin{align*}
\begin{aligned}
M &= 1/2 (I + AD^{-1}) = I - 1/2 (I - AD^{-1})\\
   &=  I - 1/2 D^{1/2} (I - D^{-1/2} A D^{-1/2}) D^{-1/2}\\
   &= I - \frac{1}{2} D^{1/2} N  D^{-1/2}\\
   & = D^{1/2} (I - N/2)D^{-1/2}.
\end{aligned}
\end{align*}
Since $N$ is a symmetric matrix, it has a set of $n$ real eigenvalues and a
corresponding set of mutually orthogonal eigenvectors.  Throughout we let
$\omega_1 \le \omega_2 \le \ldots\le \omega_n$ denote the set of eigenvalues and
$\vec \zeta_1, \vec \zeta_2, \ldots, \vec \zeta_n$ denote the corresponding set
of eigenvectors.  It is well known that eigenvalues of $N$ are
$0 = \omega_1 \le \omega_2\le \ldots \le \omega_n \le 2$.  It is easy to verify that
$\sqrt {\vec d}=(\sqrt{d_1},\sqrt{d_2},\dotsc,\sqrt{d_n})$ is the first eigenvector $\vec \zeta_1$ of $N$ with
corresponding eigenvalue $\omega_1 = 0$.  Each of the orthogonal
eigenvectors~$\vec \zeta_i$ can be normalized to be a unit vector as
$\vec e_i = \vec \zeta_i/\norm { \vec \zeta_i}_2$
Together, these
orthogonal unit eigenvectors define an orthonormal basis for $\R^n$.  Observe
that the first unit eigenvector~$ \vec e_1$ of this orthonormal basis is
$\sqrt {\vec d}/ \sqrt{2m}$.  Also observe that the stationary distribution
$\pi$ of $M$ is equal to $D^{1/2} \vec e_1/\sqrt{2m}$.

In the following we show that for every eigenvector $\vec e_i$ with eigenvalue $\omega_i$, $D^{1/2} \vec e_i$ is a right-eigenvector of $M$
with eigenvalue $1 - \omega_i/2$.
\begin{align*}
\begin{aligned}
M (D^{1/2} \vec e_i) &= D^{1/2} (I - N/2) D^{-1/2} (D^{1/2} \vec e_i) \\
                             & = D^{1/2} (I - N/2) \vec e_i \\
                             & = D^{1/2} (\vec e_i - (N \vec e_i)/2 ) = (1 - \omega_i/2) D^{1/2} \vec e_i.
\end{aligned}
\end{align*}
It follows that the eigenvalues of $M$ lie between $0$ and $1$.\\
Observe that the stationary distribution $\pi$ of $M$ is a multiple of $D^{1/2} \sqrt {\vec d}$, the first right-eigenvector of $M$ corresponding to $\vec e_1$.
%

On a connected, graph, a lazy random walk $M$
can be viewed as a reversible, aperiodic 
Markov chain with
state space $V$ and transition matrix $M$.
\begin{definition}\label{def:cheegerconstant}
  Let $\mathcal M$ be a reversible, aperiodic Markov chain on a finite state
  space $V$ with stationary distribution $\pi$. Furthermore, let
  $\pi(S) = \sum_{v \in S} \pi(v)$.  The Cheeger constant or
  conductance~$\phi(\mathcal M)$ of the chain is defined as 
\begin{equation*}
\phi(\mathcal M) = \min_{S \subset V: \pi(S) \le 1/2} \frac{\sum_{x \in S, y
    \in V\setminus S} \pi(x) \mathcal M(y,x)}{\pi(S)}. 
\end{equation*}
\end{definition}
Here $\mathcal M(y,x)$ is the probability of moving to state $y$ from state $x$ in one step.

The definition of the lazy random walk matrix $M$ and the
the fact that the stationary distribution $\pi$ of our lazy random walk is $\vec d/2m$ 
together imply that the
Cheeger constant~$\phi(M)$ (henceforth,~$\phi_*$) of the walk
$M$ is
\begin{equation*}
\phi_* = \min_{U \in V, \volume(U) \le \volume(V)/2} \frac{E(U, V\setminus U) }{2 \cdot \volume(U)}.
\end{equation*}
For an $\alpha$-conductor, we get that
\begin{equation}\label{eq:cheegereofconductors}
\phi_*  = \frac{\alpha}{2}. 
\end{equation}

%% file: algorithm-arxiv2024.tex
\section{A Distributed Algorithm for Testing Conductance}
\label{sec:alg}
Given an input graph~$G = (V,E)$, a
conductance parameter~$\alpha$, and a distance parameter $\epsilon$, our
distributed conductance tester distinguishes, with probability at least   $2/3$, between
the case where~$G$ is an $\alpha$-conductor and the case where~$G$ is
$\epsilon$-far from being an $\Omega(\alpha^2)$-conductor.
A key technical lemma from~\cite{LiP15} implies that, if $G$ is~$\epsilon$-far
from being an $\Omega(\alpha^2)$-conductor, then there exists a low-conductance
cut $(S, \comp S)$ such that~$\volume(S) \ge \epsilon m /10 $.  We build on
this lemma to show, using spectral methods (see Lemma~\ref{lem:trapprob2} and
Corollary~\ref{cor:highcountTrapT}), that there exist a set of \emph{sticky}
vertices with high enough volume in~$S$. Recall that a vertex~$x$ in a
low-conductance set~$S$ is sticky if there exists a large enough subset
$T \subset S$ such that~$x \in T$ and a short random walk starting from $x$
\emph{ends} in $T$ with a sufficiently high probability.  Intuitively, random
walks starting from sticky vertices tend to stick to a small region around
them.  This leads to some vertex in the graph receive more than their
\emph{fair share} of number of walks.  On the other, hand if a graph is a good
conductor, then the random walks from anywhere in the graph mix very
quickly. This ensures that the fraction of the total number of walks received
by any vertex in the graph is proportional to its degree.

We randomly sample a set $Q$ of $\Omega(1/\epsilon)$
\emph{source} vertices (each vertex $v \in V$ is sampled with probability proportional to its degree). 
Then we run a certain number of short random walks from each source in $Q$.  
Since a large part of the volume of our high-volume, low-conductance
set $S$ consists of only sticky vertices
 in a bad conductor, some vertex in $Q$ will be sticky with
sufficiently large probability. We use
the number of walks received by each vertex from a specific source as a test
criteria.
We implement sampling of the set $Q$ by having each vertex $v$ sample itself by
flipping a biased coin with probability $5000 \cdot \degree(v)/(2\epsilon m)$. 
Therefore, we get that $|Q| = 5000/\epsilon$ in expectation.
It follows from Chernoff bound that the probability of~$Q$ having more than~$5500/\epsilon$
vertices is at most $e^{-23/\epsilon} \le e^{-23}$.  Then, we perform $K$
random walks of length $\ell$ starting from each of the chosen vertices
in~$Q$. The exact values of these parameters are specified later in the sequel.
The pseudocode of the algorithm is presented in Algorithm~\ref{alg:disttester}.  At any
point before the last step of random walks, each vertex $v \in V$ contains a
set $W$ of tuples~$(q, count, i)$, where $count$ is the number of walks of
length~$i$ originating from source $q$ currently stationed at $v$. All these
walks are advanced by one step (for~$\ell$ times) by invoking
Algorithm~\ref{alg:moveOneStep}.  At the end of the last step of the walks,
Algorithm~\ref{alg:moveOneStep} outputs a set of tuples~$C_v$. Each tuple in
$C_v$ is of the form $(q,count)$, where $count$ is the total number
of~$\ell$-step walks starting at~$q$ that ended at~$v$.  Then, in
Algorithm~\ref{alg:disttester}, processor at vertex~$v$ goes over every tuple
$(q,count)$ in $C_v$ (see Lines~\ref{line:checkCount1} to~\ref{line:checkCount2} of Algorithm~\ref{alg:disttester}),
and if the $count$ value of any of them is above a
pre-defined threshold~$\thresholdCount$, it outputs \emph{Reject}.  If none of
the tuples have their $count$ value above threshold, it outputs \emph{Accept}.
The exact value of $\thresholdCount$ is specified later in the sequel.

When advancing the walks originating at a source $q \in Q$ by one step, 
we do not send the full trace of every random walk.
Instead, for every source $q \in Q$, every vertex~$v \in V$ only sends a
tuple $(q, k, i)$ to its neighbour $w$ indicating that $k$ random walks
originating at~$q$ have chosen~$w$ as their destination in their~$i$th step.
Since the size of $Q$ is constant with high enough probability, we will not
have to send more than a constant number of such tuples on each edge.  Moreover,
each tuple can be encoded using $O(\log n)$ bits (given the values of parameters $\ell$ and $K$ specified in the sequel). 
Hence, we only communicate
$O(\log n)$ bits per edge in any round with high probability. 
To ensure no congestion, 
we check the length of every message (Algorithm~\ref{alg:moveOneStep}: 
lines~\ref{lin:checkMessageLengthBegin} -~\ref{lin:checkMessageLengthEnd}
).  
If a message appears too
large to send, we simply output \emph{Reject} on the host vertex and abort the
algorithm.  Note that the number of tuples we ever have to send along any edge
is upper bounded by $|Q|$ and $|Q| \le 5500\epsilon$, with probability at least
$1 - e^{-23}$.  Therefore, we may rarely abort the algorithm before completing
the execution of the random walks.  If that happens, then the probability of
accepting an $\alpha$-conductor is slightly reduced.  Hence the
following observation follows:
\begin{observation}\label{obs:rejectForLargeMessage}
  Algorithm~\ref{alg:disttester} rejects an $\alpha$-conductor due to
  congestion with probability at most~$e^{-23}$.
\end{observation}

We set the required parameters of
Algorithm~\ref{alg:disttester} as follows:
\begin{itemize}
\item the number of walks $K = 2m^2$, 
\item the length of each walk $\ell = \frac{32}{\alpha^2} \log n$
\item the rejection threshold for vertex $v \in V$, $\thresholdCount_v = m\cdot\degree(v)\cdot(1+2 n^{-1/4})$.
\end{itemize}

\begin{algorithm}[th!]
{\small
\caption{Distributed algorithm running at vertex $v$ for testing
  conductance.}
\label{alg:disttester}
\begin{algorithmic}[1]
\Procedure{Distributed-Graph-Conductance-Test}{$G,\epsilon, \alpha, \ell,K$}
\\  
\Comment{The algorithm performs $K$ random walks of length $\ell$ from a set $Q$ of
  $\Theta(1/\epsilon)$ starting vertices, where every starting vertex is
  sampled randomly from $V$.} 
\State $\ell$ : The length of each random walk
\State $K$ : The number of walks
\State $W_v$ : Set of tuples $(q,count,i)$ \Comment{where $count$ is
  the number of walks originating at source $q$ currently stationed at
  $v$} 

\State $C_v$ : Set of tuples $(q,count)$ \Comment{where $count$ is
  the total number of $\ell$ step walks starting at~$q$ that ended at $v$} 
 \State $\thresholdCount_v$ : maximum number of $\ell$-length walks $v$ should see from a given source on an $\alpha$-conductor.
\State Flip a biased coin with probability $p = 5000 \degree(v)/(\epsilon 2m)$ to decide
whether to start $K$ lazy random walks. 
\State If chosen, initialise~$W_v$ as $W_v \gets \{ (v,K,0) \}$.
\State Call Algorithm~\ref{alg:moveOneStep} for $\ell$ synchronous rounds.
\While{there is some tuple~$(q, count)$ in~$C_v$}\label{line:checkCount1}
	\If{ $count > \thresholdCount_v$ } \Comment{Received too many walks from $q$.}
	    \State Output \emph{Reject} and stop all operations.
	 \Else
	      \State Remove $(q, count)$ from $C_v$
	 \EndIf
\EndWhile\label{line:checkCount2}
\State Output \emph{Accept}
\EndProcedure
\end{algorithmic}
}
\end{algorithm}
\begin{algorithm}[th!]
{\small
\caption {Algorithm for moving random walks stationed at $v$ by one step.}
\label{alg:moveOneStep}
\begin{algorithmic}[1]
\Procedure{Move-Walks-At-$v$}{}
\State $W_v$ : Set of tuples $(q,count,i)$ \Comment{where $count$ is the
  number of walks originating at source $q$ currently stationed at $v$ just after their $i$th step.}. 
\State $D_v$ : Set of tuples $(q,count,dest)$ \Comment{where $count$ is
  the number of walks starting at  $q$ that are to be forwarded to $dest$.} 
\State $C_v$ : Set of tuples $(q,count)$ \Comment{where $count$ is the
  total number of walks starting at $q$ that have $v$ as their final
  destination or endpoint.} 
\State $D_v \gets \emptyset$. 
%
\While{there is some tuple~$(q,k,i)$ in~$W_v$}
\If{$i \neq L$} \Comment{If not the last step, process the next set of destinations.}
\State Draw the next set of destinations for the $k$ walks and update the set $D_v$.
\State Remove $(q,k,i)$ from $W_v$
\LeftComment{If last step of the walks, update how many ended at $v$. }
\State Update $C_v$ to reflect the $k$ walks that ended in $v$.
\EndIf
\EndWhile

\LeftComment{Prepare the messages to be sent}
\While{there is some tuple~$(q,count,dest)$ in $D_v$}
\State Add tuple~$(q,count,i+1\rangle$ to the message to be sent to $dest$
\EndWhile
\LeftComment{Check each message for length}
 \State For each message $M$ to be sent\label{lin:checkMessageLengthBegin}
	 \If {the number of tuples in $M > 5500/\epsilon$}
	 	 \State Output \emph{Reject} and stop all operations.\label{lin:checkMessageLengthEnd}
	  \EndIf      

\State Send all the messages to their respective destinations.

 \LeftComment{Process the messages received}
 \State For each source $s$ from which a total of $s_{count}$ walks are received, \\
 add tuple $(s, s_{count},i +1)$ to $W_v$ 
\EndProcedure
\end{algorithmic} 
}
\end{algorithm}
\subsection{Analysis of the Algorithm}\label{sec:combproof}
The main idea behind our algorithm is that, in a bad conductor, a random walk
would converge to the stationary distribution more slowly and would initially
get trapped within sets of vertices with small conductance. 
We provide a lower bound on the probability of an $\ell$-step random walk
starting from a vertex chosen at random (with probability proportional to its degree) from a subset~$T$ of 
a low-conductance set~$S$ finishing at some vertex in~$T$.
\begin{definition}\label{def:trapprpb}
  For a set~$T \subseteq V$, and a vertex $u \in T$, let
  $\trap(u, T, \ell)$ (henceforth \emph{trap probability}) denote the
  probability of an $\ell$-step random walk starting from $u \in T$ finishing
  at some vertex in~$T$. When the starting vertex is chosen at random
  from $T$ with probability proportional to its degree, we denote by~$\trap(T,
  \ell)$ the average trap probability (weighted by vertex degrees) over
  set $T$:
  $\trap(T, \ell) = \frac1{\volume(T)} \sum_{u \in T} \degree(u)\cdot \trap(u, T, \ell)$.
\end{definition}
Given a set $S$ of conductance at most $\delta$ and $T \subseteq S$, 
we establish a relationship between the average trap probability
$\trap(T, \ell)$ and conductance~$\delta$ of~$S$ in the next two
lemmas. We first consider the case $T=S$ in Lemma~\ref{lem:trapprob1}, and then
obtain a bound when~$T$ is a subset (of sufficiently large volume) of~$S$ in Lemma~\ref{lem:trapprob2}.
\begin{lemma}\label{lem:trapprob1}
Consider a set $S \subseteq V$ such that the cut $(S, \comp{S})$ has
conductance at most~$\delta$. Then, for any integer $\ell >0$, the following holds
\[
	\trap(S, \ell) \ge \frac{\volume(S)}{2m} + \left(\frac{5}{6} - \frac{\volume(S)}{2m}   \right) (1-3\delta)^{\ell}.
\]
\end{lemma}
\begin{proof}
  Let $\ind{S}$ denote the $n$-dimensional indicator vector of set $S$.
  We pick a source vertex $v \in S$ with probability $\degree(v)/\volume(S)$, 
  where $\degree(v)$ is degree of $v$ and $\volume(S)$ is the sum of degrees of vertices in set $S$. 
  This defines an initial probability distribution denoted by $\vec{p}^{\,0}_S$ on the vertex set $V$, 
  where, $\vec{p}^{\,0}_S(v) = d(v)/\volume(S)$ for $v \in S$ and $\vec{p}^{\,0}_S = 0$ for $v \notin S$.
  Note that $\vec{p}^{\,0}_S = \frac{1}{\volume (S)} D \ind{S}$,
  where,  $D$ is the diagonal degree matrix of $G$.
  Denote by $M$ the transition matrix of a lazy random walk on $G$.
 The endpoint probability distribution~$\vec{p}^{\,\ell}_S$ of an $\ell$-step lazy
 random walk on $G$ starting from a vertex chosen from $S$ according to $\vec{p}^{\,0}_S$ is
\[\vec{p}_S^{\,\ell} = M^{\ell} \vec{p}^{\,0}_S = (1/\volume (S)) M^{\ell} D \ind{S}.\]
Recall that $M$ can be expressed in terms of the normalized Laplacian matrix $N = I - D^{-1/2} A D^{-1/2}$ of $G$ as
$M = D^{1/2} (I - N/2) D^{-1/2}$. (See Section~\ref{sec:Prelim}.)

It follows therefore that
$ \vec{p}_S^{\,\ell} = (1/\volume (S)) \left(D^{1/2} \left(I - N/2 \right) D^{-1/2} \right)^{\ell} D \ind{S}$.

The trap probability~$\trap(S,\ell)$ of an $\ell$-step lazy random walk
starting from a random vertex in~$S$  picked according to $\vec{p}^{\,0}_S$ can
be expressed as the inner product of vectors~$\vec{p}_S^{\,\ell}$ and~$\ind{S}$:
\begin{align*}
	\begin{aligned}
 	\trap(S, \ell) &= \frac{1}{\volume (S)} \trind{S} M^{\ell} D \ind{S} =   \frac{1}{\volume (S)} \trind{S} \left(D^{1/2} \left(I - N/2 \right) D^{-1/2}\right)^{\ell}    D \ind{S}\\
                    	  & =   \frac{1}{\volume (S)} \trans{(D^{1/2} \ind{S})} \left(I - N/2 \right)^{\ell} (D^{1/2} \ind{S}).
	\end{aligned}
\end{align*}
Recall that $0 = \omega_1 \le \omega_2 \le \cdots \le \omega_n < 2$ are the
eigenvalues of $N$ and $\vec e_1, \vec e_2, \ldots, \vec e_n$ denote the
corresponding unit eigenvectors.
We can express $D^{1/2} \ind{S}$ in the orthonormal basis defined by the eigenvectors of $N$ as
$D^{1/2} \ind{S} = \sum_i \alpha_i \vec e_i$. It follows that
\begin{equation}\label{eq:innerp}
	 \sum_i \alpha_i^2 = \langle D^{1/2} \ind{S},  D^{1/2}\ind{S} \rangle= \volume(S).
\end{equation}
Taking the
quadratic form of $N$ for vector $D^{1/2} \ind{S}$, we get
\begin{align*}
\begin{aligned}
\trans{(D^{1/2} \ind{S})} N (D^{1/2} \ind{S}) &= \trans{(D^{1/2} \ind{S})} I  (D^{1/2} \ind{S})  -  \trans{(D^{1/2} \ind{S})} (D^{-1/2} A D^{-1/2}) (D^{1/2} \ind{S}) \\ 
                    &=\volume (S) - \trans{(D^{1/2} \ind{S})} (D^{-1/2} A D^{-1/2}) (D^{1/2} \ind{S})
                     = \volume (S) - \trind{S} A \ind{S}.
\end{aligned}
\end{align*}
Note that the term $ \trind{S} A \ind{S}$ corresponds to the number of edges in $S \times S$.
Therefore. it follows that
$$
	\trans{(D^{1/2} \ind{S})} N (D^{1/2} \ind{S}) = E(S, \comp{S}).
$$
Since the conductance of the cut $(S, \comp{S})$ is at most $\delta$,
we have that
\begin{equation}\label{eq:quadratic1}
	\trans{(D^{1/2} \ind{S})} N (D^{1/2} \ind{S})=E(S, \comp{S}) \le \delta\cdot \volume (S).
\end{equation}
Expressing $D^{1/2} \ind S$ as $\sum_i \alpha_i \vec e_i$, the quadratic form
of $N$ for $D^{1/2} \ind S$ can also be written as 
\begin{equation}\label{eq:quadratic2}
 	\trans{(D^{1/2} \ind{S})} N (D^{1/2} \ind{S}) = \trans{ (\sum_i \alpha_i \vec e_i)} N  (\sum_i \alpha_i \vec e_i) = \sum_i \alpha_i^2 \omega_i.
 \end{equation}
Combining from Eq.s~\eqref{eq:innerp}, \eqref{eq:quadratic1} and~\eqref{eq:quadratic2}, we get that
\begin{equation}\label{eq:quadratic3}
  \trans{(D^{1/2} \ind{S})} (I - N/2) (D^{1/2} \ind{S}) = \sum_i \alpha_i^2 - \frac{1}{2} \sum_i \alpha_i^2 \omega_i \ge \volume(S) - \frac{\delta}{2} \volume(S).
\end{equation}
Recall that $0 = \omega_1 \le \omega_2 \le \ldots \le \omega_n < 2$ are the
eigenvalues of $N$ and let $\vec e_1, \vec e_2, \ldots, \vec e_n$ be the
corresponding unit eigenvectors. Correspondingly, we can define
a set of eigenvalues $1= \lambda_1 \ge \lambda_2 \ge \ldots \ge \lambda_n > 0$
and the same set of eigenvectors  $\vec e_1, \vec e_2, \ldots, \vec e_n$ for $I - N/2$.
Notice that for each $i$, $\lambda_i = 1 - \omega_i/2$.
With this translation of eigenspace, we get that 
\begin{equation*} 
\trans{(D^{1/2} \ind{S})} (I -N/2) (D^{1/2} \ind{S}) = \sum_i \alpha_i^2 \lambda_i .
\end{equation*}
We call the quantity
$\sum_i \alpha_i^2$ the \emph{coefficient sum of the eigenvalue set}.  We also
call an eigenvalue~$\lambda_i$ of $I - N/2$ (and the corresponding eigenvector~$\vec e_i$)
\emph{heavy} if $\lambda_i \ge 1 - 3\delta$.  We denote by $H$ the index set of
the heavy eigenvalues and let $\comp{H}$ be the index set of the rest.  Since
$\sum_i \alpha_i^2 \lambda_i \ge  (1 - \delta/2) \volume (S)$ is large for a set with small
conductance, we expect many of the coefficients $\alpha_i^2$ corresponding to
heavy eigenvalues to be large.  This would slow down the convergence of the
random walk and make the trap probability for our low-conductance set $S$
large. The following 
claim
establishes a lower bound on the contribution of the
index set $H$ to the coefficient sum.
\begin{claim}\label{cl:xvalue}
For~$\{\alpha_i\}_i, H$, and~$s$ as defined above,
$$ \sum_{i \in H} \alpha_i^2 \ge \frac{5}{6} \volume(S).$$
\end{claim}


\begin{proof}[Proof of the claim]
   Let $x$ denote the coefficient sum of the set $H$ of heavy eigenvalues: that
   is,  $ x \eqdef \sum_{i \in H} \alpha_i^2$.
 The following expression follows by the definition of the set $H$ and its
 coefficient sum $x$:
 \begin{equation*}
   \sum_i \alpha_i^2 \lambda_i = \sum_{i \in H} \alpha_i^2
   \lambda_i + \sum_{i \in \comp{H}} \alpha_i^2 \lambda_i
    \le x  + (\sum_i \alpha_i^2 -x ) (1-3\delta).
 \end{equation*}
 The second inequality above follows by upper bounding every~$\lambda_i$ with
 $i\in H$ by $1$ and every~$\lambda_i$ with~$i \in \comp{H}$ by~$1-3\delta$.
 Recall that $\sum_i \alpha_i^2 = \langle D^{1/2} \ind{S}, D^{1/2} \ind{S} \rangle = \volume(S)$ and we
 just proved that $\sum_i \alpha_i^2 \lambda_i \ge  (1 - \delta/2) \volume(S)$
 in~\eqref{eq:quadratic3}.  It follows that
 \begin{equation*}
   (1 - \delta/2) \volume(S)  \le \sum_i \alpha_i^2  \lambda_i   \le x +
   (\volume(S)-x)(1-3\delta).   
 \end{equation*}
 Rearranging the inequality above, we get that $x \ge 5 \volume(S)/6$.
\end{proof}
Next, we use 
Claim~\ref{cl:xvalue}
to obtain a lower bound on the average trap
probability of set $S$ in terms of the conductance of the cut $(S, \comp{S})$.
\begin{align*}
	\begin{aligned}
  \trap(S, \ell) &= \frac{1}{\volume(S)} \trans{(D^{1/2} \ind{S})} (I - N/2)^{\ell} (D^{1/2} \ind{S}) \\
  		     &= \frac{1}{\volume(S)} \trans{(\sum_i \alpha_i \vec e_i )} (I - N/2)^{\ell} (\sum_i \alpha_i \vec e_i)\\
  	              &= \frac{1}{\volume(S)} \trans{(\sum_i \alpha_i  \vec e_i )} (\sum_i \alpha_i
 			 \lambda_i^{\ell} \vec e_i) = \frac{1}{\volume(S)} \sum_i \alpha_i^2
                           \lambda_i^{\ell}.
  \end{aligned}
\end{align*}
Further, focusing on the contribution of the index set~$H$ to the trap
probability,
\begin{align}
  \trap(S, \ell)
  &= \frac{1}{\volume(S)} \sum_i \alpha_i^2 \lambda_i^{\ell } \ge \frac{1}{\volume(S)} \sum_{i\in H} \alpha_i^2 \lambda_i^{\ell }\nonumber\\
  &=  \frac{1}{\volume(S)} (\alpha_1^2  \lambda_1 +  \sum_{i \in
    H\setminus \{1 \}} \alpha_i^2  \lambda_i^{\ell} )\nonumber\\ 
  & \ge \frac{1}{\volume(S)}\left(\alpha_1^2 + ( 5\volume(S)/6   -
    \alpha_1^2) \left(1-3\delta\right)^{\ell} \right). \label{eq:trap}
\end{align}
The last inequality follows by the definition of a heavy eigenvalue and by
Lemma~\ref{cl:xvalue} 
(we have that~$\sum_{i \in H} \alpha_i^2 \ge 5\volume(S)/6$).
By definition, $\lambda_1 =1$ and
$\vec e_1 = \sqrt{\vec d}/\sqrt{2m} $, where, $\vec d$ is the vector of vertex degrees. It follows that
  $\alpha_1 = \langle D^{1/2} \ind{S}, \vec e_1 \rangle = \volume (S)/\sqrt{2m}$.
Plugging in the values of $\alpha_1$ in~\eqref{eq:trap}, we get
\begin{align*}
          \trap(S, \ell) &\ge \frac{1}{\volume(S)} \left( \frac{(\volume(S))^2}{2m} +
            \left(\frac{5\volume(S)}{6} - \frac{(\volume(S))^2}{2m} \right) \left(1 -
              3\delta \right)^{\ell} \right) \\ 
          & = \frac{\volume(S)}{2m} + \left( \frac{5}{6} - \frac{\volume(S)}{2m}
          \right) (1 - 3\delta)^{\ell}.
\end{align*}
\end{proof} 
Next lemma states that every subset $T \subset S$ of large enough volume has high trap probability.  
 More specifically, we prove a lower bound on the probability of
 an $\ell$-step random walk starting from a vertex chosen at random (with probability proportion to its degree) from $T$ finishing at some vertex in $T$.
\begin{lemma}\label{lem:trapprob2}
 Consider sets $T \subseteq S \subseteq V$, such that the
cut $(S, \comp{S})$ has conductance at most $\delta$ 
and that~$\volume(T) = (1 -\eta) \volume(S)$ for some $0 < \eta < 5/6$, 
then for any integer~$\ell>0$, there exists a vertex $v \in T$ such that 
\begin{equation}\label{eq:hightrap}
	 \trap(v, T, \ell) \ge \frac{\volume(T)}{2m} + \left(\frac{5}{6} ( 1-
             \sqrt{(6\eta)/5} )^2    - \frac{\volume(T)}{2m} \right)
         (1-3\delta)^{\ell} .
\end{equation}
\end{lemma}
\begin{proof}
  Let $\ind{S}$ and $\ind{T}$ denote the $n$-dimensional indicator vectors of
  sets $S$ and $T$, respectively. As in the proof of Lemma~\ref{lem:trapprob1}, we express
  $D^{1/2}\ind{S}$ and $D^{1/2}\ind{T}$ in the orthonormal basis defined by the eigenvectors
  of the normalized Laplacian matrix $N$ as $ D^{1/2}\ind{S} = \sum_{i} \alpha_i \vec e_i$ and
  $D^{1/2}\ind{T} = \sum_{i} \beta_i \vec e_i$.
  Since the conductance of the cut $(S, \comp{S})$ is at most $\delta$, Claim~\ref{cl:xvalue}
from Lemma~\ref{lem:trapprob1} holds.
We have that
          $\sum_{i} \alpha_i^2 \ge \sum_{i \in H} \alpha_i^2 \ge \frac{5}{6} \volume(S)$.  
By the definition of $D^{1/2}\ind{S}$ and $D^{1/2} \ind{T}$, we have
  \[
  	\norm{D^{1/2}\ind{S} - D^{1/2} \ind{T} }_2^2 = \volume(S) - \volume(T) = \volume(S) - (1-\eta) \cdot \volume(S)  = \eta \cdot \volume(S).
  \]
  Furthermore, the following follows from the expression of $D^{1/2}\ind{S}$
  and $D^{1/2}\ind{T}$ in terms of the eigenvectors of $I - N/2$:
  \[
  	\eta \cdot \volume(S) =\norm{D^{1/2}\ind{S} - D^{1/2} \ind{T}}_2^2 = \sum_{i} (\alpha_i - \beta_i)^2. 
  \]
  Applying the triangle inequality,
  $\norm{\vec a - \vec b} \ge \norm{\vec a} - \norm{\vec b}$ and upper
  bounding $\sqrt{\sum_{i\in H} (\alpha_i - \beta_i)^2}$ by
  $\sqrt{\sum_{i} (\alpha_i - \beta_i)^2} = \sqrt{\eta \cdot \volume(S)}$, we get that
\begin{align*}
 	\begin{aligned}
          \sum_{i \in H} \beta_i^2 \ge & \left( \sqrt{\sum_{i \in H}
              \alpha_i^2} - \sqrt{\sum_{i \in H} \left(\alpha_i - \beta_i
              \right)^2} \right)^2\\ 
          \ge & \left( \sqrt{(5 \volume(S))/6 } - \sqrt{\eta \cdot \volume(S) } \right)^2 =
          \frac{5}{6} \volume(S) \left(1 - \sqrt {(6 \eta)/5 }\right)^2.
 	\end{aligned}
\end{align*}
Reasoning as in Lemma~\ref{lem:trapprob1} and applying $\lambda_i \ge (1- 3
\delta)$, for all~$i \in H$, we can bound the average trap probability over set
$T$ as  
\begin{align}\label{eq:trapT}
	\begin{aligned}
          \trap(T, \ell) \ge& \frac{1}{\volume(T)} \sum_{i \in H} \beta_i^2
          \lambda_i^{\ell} = \frac{1}{\volume(T)} (  \beta_1^2 \lambda_1 +  \sum_{i
              \in H \setminus \{1\}} \beta_i^2 \lambda_i^{\ell} )\\ 
          \ge& \frac{1}{\volume(T)} \left ( \frac{(\volume(T))^2}{2m} +  \left(\frac{5\volume(S)}{6} \left(
                1- \sqrt{(6 \eta)/5}\right)^2 - \frac{(\volume(T))^2}{2m}   \right
            ). (1 - 3 \delta)^{\ell} \right)\\ 
          =& \frac{\volume(T)}{2m} + \left(\frac{5}{6} \cdot \frac{\volume(S)}{\volume(T)}  \left ( 1-
              \sqrt{(6 \eta)/5}\right)^2 - \frac{\volume(T)}{2m} \right) (1-3\delta)^{\ell}.
	\end{aligned}
\end{align}
The second-to-last inequality follows from the definition of the first
eigenvalue, eigenvector pair of matrix $I - N/2$.  By definition,
$\lambda_1 =1$ and $\vec e_1 = d^{1/2}/\sqrt{2m}$. Therefore, we have
that~$\beta_1 = \langle D^{1/2}\ind{T}, \vec e_1 \rangle = \volume(T)$.  Since
$T \subseteq S$, $\volume(S)/\volume(T) \ge 1$. The lemma follows by
substituting $\volume(S)/\volume(T) $ by its lower bound $1$
in~\eqref{eq:trapT}.
Thus it follows that there exists a sticky vertex~$v \in T$ such that
 $ \trap(v, T, \ell)$ is given by~\eqref{eq:hightrap}.
 \end{proof}

Lemma~\ref{lem:trapprob2} implies the following corollary.
\begin{corollary}\label{cor:highcountTrapT}
  Consider a set $S \subset V$, such that the cut $(S, \comp{S})$ has
  conductance at most~$\delta$.  
  Given any~$0 < \eta < 5/6$ and integer
  $\ell >0$, there exist a set of volume at least $\eta \cdot \volume(S)$, such that every 
  vertex in this set is sticky. In other words, for every vertex $v$ is this set,
  there exists~$T \subseteq S$ of volume $ \volume(T)=(1-\eta) \cdot \volume(S)$ such that
  $\trap(v, T, \ell$) is given by~\eqref{eq:hightrap}.
\end{corollary}
\begin{proof}
  Let $P \subseteq S$ denote the set of all the vertices for which
  eq~\eqref{eq:hightrap} holds for some set~$T$.  One can extract $P$ using
  the following iterative procedure.  To begin with, we pick an arbitrary
  subset $T$ of volume $(1-\eta) \cdot \volume(S)$ from $S$.  By Lemma~\ref{lem:trapprob2},
  there exists a vertex $v$ in $T$ with the desired trap probability.  We
  remove $v$ from $S$ and add it to $P$. Let $R$ be the set of remaining
  vertices of $S$.  We then extract another subset $T$ of volume $(1-\eta) \cdot \volume(S)$
  from $R$ and this process continues until we do not have sufficient volume
   left in $R$.  It is easy to see that~$\volume(R) < (1-\eta) \cdot \volume(S)$ when this
  process ends. This implies that the volume of $P$ is at least $\eta \cdot \volume(S)$.
\end{proof}

We build on the
following combinatorial lemma from~\cite{LiP15}:
\begin{lemma}[Lemma~9 of~\cite{LiP15}]\label{lem:combPan}
  Let $G = (V, E)$ be an $m$-edge graph. If there exists a set~$P\subseteq V$
  such that $\volume (P) \le \epsilon m/10$ and the subgraph
  $G[V \setminus P ]$ that is induced by the vertex set $V\setminus P$ has conductance at least~$\phi'$, then there exists an
  algorithm that modifies at most $\epsilon m$ edges of $G$ to get a
  graph $G' = (V, E')$ with conductance at least $\phi/'3$.
\end{lemma}
In the following lemma, we show the existence of a high volume set $A$ with low enough conductance in a graph that is far from being a good conductor.
\begin{lemma}\label{lem:highVolSet}
Let $G= (V,E)$ be an $n$-vertex, $m$-edge graph such that $G$ is $\epsilon$-far from having conductance at least $\alpha^2/2880$,
then there exists a set $A \subseteq V$ such that $\volume(A) \ge \epsilon m /10$ and conductance of cut $(A, \comp A)$ is at most
$\alpha^2/960$.
\end{lemma}
\begin{proof}
We build the set $A$ iteratively from sets $A_1, A_2, \ldots, A_i,\ldots$ as follows.
We begin with $A=A_0 = \emptyset$ and $\comp A_0 = V$.
In iteration $i$, we look for a cut $(A_i, \comp A_i)$ in $\comp A_{i-1}$ such that
$\volume(A_i) \le \volume(\comp{A_i})$ and conductance of cut $(A_i, \comp A_i)$
is at most $\alpha^2/960$. We set $A = A \cup A_i$ and $\comp A_i = \comp A_{i-1}\setminus A_i$. 
This process continues until we can find such a cut in $A_{i-1}$.
It follows from the procedure that the conductance of the cut $(A, \comp A)$ when we stop is
at most $\alpha^2/960$.
Next, we claim that the volume of the set $A$ is at least $\epsilon m /10$. 
Let us assume for contradiction that $\volume (A) < \epsilon m/10$.
Then, by Lemma~\ref{lem:combPan}, with at most $\epsilon m$ edge updates,
we can get a graph $G'$ with conductance at least $\alpha^2/2880$ from $G$.
But we know from the statement of the lemma that $G$ is $\epsilon$-far from
having conductance at least $\alpha^2/2880$.
Hence, the claim follows by contradiction.
\end{proof}

Finally, we need the following classical relation between the conductance or
Cheeger constant of a Markov chain and its second largest eigenvalue.
\begin{theorem}[\cite{AlonMilman,Alon86eigenvalues,Dodziuk}]
  \label{thm:cheegerconstant}
  Let $P$ be a reversible lazy chain (i.e., for all~$x$, $P(x,x) \ge 1/2$) with
  Cheeger constant $\phi_*$.  Let $\lambda_2$ be the second largest eigenvalue
  of $P$. Then,
$\frac{\phi_*^{2}}{2} \le 1- \lambda_2 \le 2\phi_*$.
\end{theorem}
We can now state our main theorem.
\begin{theorem}
\label{thm:mainthm}
  For an input graph $G = (V,E)$,
  parameters~$0 < \alpha < 1$ and~$\epsilon > 0$, 
the distributed algorithm described in Section~\ref{sec:alg} 
\begin{itemize}
\item outputs Accept, with probability at least $2/3$, on
  every vertex of $G$ if $G$ is an $\alpha$-conductor.
\item outputs Reject, with probability at least $2/3$, on at least one
  vertex of $G$ if $G$ is $\epsilon$-far from any~$(\alpha^2/2880)$-conductor. 
\end{itemize}
The algorithm uses $O(\log n/\alpha^2)$ communication rounds.
\end{theorem}
\begin{proof}
  Let us start by showing that, with high enough probability, the algorithm outputs
  \emph{Accept} on every vertex if $G$ is an $\alpha$-conductor.  By
  Observation~\ref{obs:rejectForLargeMessage}, we may reject $G$ and abort the
  algorithm due to congestion with probability at most $e^{-23}$. For now, let
  us assume that this event did not occur. Denote by $\lambda_2$ the second
  largest eigenvalue of the lazy random walk matrix~$M$ on $G$.  It is well known
  (see, e.g.,~\cite{Sinclair}) that, for a pair $u,v \in V$,
  $\left| M^{\ell}(v,u) - \degree(v)/(2m) \right| \le \lambda^{\ell}_2 \le e^{-\ell(1-
    \lambda_2)}$.
It follows from Theorem~\ref{thm:cheegerconstant} that 
 $ \left |M^{\ell}(v,u) - \frac{\degree(v)}{2m} \right| \le e^{-\ell \phi^2_*/2} \le e^{-\frac{\ell \alpha^2}{8}}$,
where the second inequality above follows from the fact that, for a random walk
on an $\alpha$-conductor, $\phi_* = \alpha/2$ (see~\eqref{eq:cheegereofconductors}).  Thus, in an~$\alpha$-conductor, for
$\ell = (32/\alpha^2) \log n$, any starting vertex $u \in V$ and a
fixed vertex $v \in V$, we have that
\begin{equation*}
  \degree(v)/ (2m) - 1/n^4 \le M^{\ell}(v,u) \le \degree(v)/(2m) + 1/n^4.
\end{equation*}
Recall that the number~$K$ of random walks and rejection threshold
$\thresholdCount_v$ for vertex~$v$ are set as $K = 2m^2$ and 
$\thresholdCount_v = m \cdot \degree(v) \cdot (1+ 2n^{-1/4})$.
Let~$X_{u,v}$ denote the number of random walks starting from $u$ that ended
in~$v$.  It follows that 
\begin{equation*}
  \mathbb E X_{u,v} = K  \cdot M^{\ell}(v,u) \le 2m^2 \cdot  \frac{\degree(v)}
  {2m} + \frac{2m^2}{n^4} \le m \cdot \degree(v) + 1.
\end{equation*}
The random variable~$X_{u,v}$ is the sum of~$K$ independent Bernoulli trials
with success probability~$M^{\ell}(v,u)$.  Applying multiplicative Chernoff
bounds, we get the following for large enough $n$.
\begin{align*}
  	\Pr[X_{u,v} > (1+n^{-1/4})\cdot \mathbb{E} [X_{u,v}] ]
  	&< \exp(-n^{-1/2}\cdot (m\cdot \degree(v) + 1)/3)
     \le \exp(\frac{-n^{1/2}}3).
\end{align*}
The second inequality above follows from the fact that $m \cdot \degree(v) > n-1$ for a connected graph.
 Thus, each vertex~$y$ receives at most
\[ (1+n^{-1/4})\cdot \mathbb{E} [X_{u,v}] \le (1+n^{-1/4})\cdot (m \cdot \degree(v)  +
  1) < m\cdot \degree(v) + 2m/n^{1/4} \cdot \degree(v)
\]
walks from~$u$, with probability at least~$1- n^{-2}$.
Taking union bound over all $y \in V$ and all starting vertices $u$, we get
that, with high probability, our algorithm outputs \emph{Accept} on every vertex of~$G$
for every starting point if~$G$ is an~$\alpha$-conductor. 
Finally, taking the
union bound over the events that we rejected due to congestion or due to
receiving too many walks at some vertex, the claim follows.

Next, we analyse the probability of rejecting if $G$ is far from having the desired conductance.
By Lemma~\ref{lem:highVolSet}, there
exists a set~$S \subset V$, with~$\volume(S) \ge \epsilon m /10$, such that the
conductance of $S$ is at most~$\alpha^2/960$. Further applying
Corollary~\ref{cor:highcountTrapT}
with~$\eta = 5/486$ on $S$ as above, we get that
there exists a set $P$ of sticky vertices such that $ \volume(P) \ge (5\volume(S))/486$ and for every $v \in P$,
there exists a set $T \subseteq S$, $\volume(T) = (481 \volume(S)) /486$, such that
 \begin{align*}
     \trap(v, T, \ell) &\ge \frac{\volume(T)}{2m} + \left(\frac{160}{243} - \frac{\volume(T)}{2m}\right)
   \left(1- \frac{\alpha^2}{320} \right)^{\ell} \\
   &\ge \frac{\volume(T)}{2m} +  \left(\frac{160}{243} - \frac{\volume(T)}{2m}\right)
   \cdot{e}^{-\alpha^2\ell/160}
   \end{align*}
 where the last inequality follows from that $1-x>e^{-2x}$, for $0\le x < 1/2$,
 provided that~$\alpha^2/320 \le 1/2$.  
 For $\ell = (32/\alpha^2) \log n$ and
 $\volume(T) \le \volume(S) \le \volume(V)/2 = m$, we get that, for every $v \in P$,
  $\trap (v,T, \ell) \ge \frac{\volume(T)}{2m} +  \frac{77}{486} \cdot \frac{1}{n^{1/5} }$.

Let us assume that a vertex $u \in P \subset A$ is picked as the starting
vertex of $K = 2m^2$ random walks in~$G$.  By
Corollary~\ref{cor:highcountTrapT}, a set $T$ with
$\volume(T) = (481\volume(S)) /486$ with
$\trap(v, T, \ell) \ge \volume(T)/2m + 77/486 \cdot n^{-1/5}$ will
exist, for every $v \in P$. 
Also note that $\volume(T) \le \volume(S) \le \volume(V) =2 m$.
For some appropriate constant $c_1$, $\volume(T) = c_1 \cdot  m = \Theta(m)$.
Further, let~$Y_{u, T}$ be the number of walks
that ended in the set $T$ (corresponding to $u$ as in
Corollary~\ref{cor:highcountTrapT}) after $\ell$ steps.  It follows that
\begin{equation*}
  \mathbb E Y_{u,T} \ge K\cdot \left( \frac{\volume(T)}{2m} + \frac{77}{486}n^{-1/5} \right) 
  \ge m \cdot \volume(T) +
  \frac{154}{486} \cdot \frac{m^2}{n^{1/5}} . 
\end{equation*}
Let $c_2 =  \frac{77}{486}$.
By an application of Chernoff
bound,we get
\begin{align*}
\lefteqn{\Pr\left[ Y_{u,T} < \left(1 - 3(m\cdot \volume(T))^{-1/2} \right) \mathbb{E}[Y_{u,T}] \right]}\\
&\quad < \exp\left( - 4(m\cdot \volume(T))^{-1} \cdot \mathbb{E}[Y_{u,T}]\right)\\
&\quad \le \exp\left( -4 (m\cdot \volume(T))^{-1} \cdot \left(\volume(T) \cdot m + 2 c_2 m^2 n^{-1/5} \right)
                                                    \right)\\
& \quad \le \exp\left(-4\cdot (1+ \Theta(n^{-1/5}) \right) < 1/10.
\end{align*}
With probability at least $9/10$,the total number of walks received by set $T$ is
\begin{align*}
	\begin{aligned}
 	&\ge (1 -  3(m \cdot \volume(T))^{-1/2}) \left(m \cdot \volume(T) + 2 c_2 {m^2/n^{1/5}}\right)\\
	&\ge m \cdot \volume(T) + 2 c_2 {m^2/n^{1/5}} - 3\sqrt m \sqrt
   {\volume(T)} - 6 c_2 m^{3/2} (\volume(T))^{-1/2} n^{-1/5} .
 \end{aligned}
\end{align*}
 The number of walks received by any vertex $v \in T$ is minimum when
the walks within $T$ have mixed well 
reaching their stationary distribution with respect to 
$T$. 
It follows that the number of walks received by a vertex $v \in T$
is at least
$$
       \degree(v)\cdot \left( m   + 2 c_2 \frac{m^2}{n^{1/5} \cdot \volume(T)} 
                                 - 3 \frac{\sqrt m \cdot \sqrt{\volume(T)}} {\volume(T)} 
                                  - 6 c_2 \frac{m^{3/2}} {(\volume(T))^{3/2} n^{1/5} }
                                    \right).
$$
Recalling that $\volume(T) = c_1 m$, the expected number of walks
received by a vertex $v \in T$ is at least
\begin{equation*}
  \degree(v) \cdot \left( m   + 2 \frac{c_2}{c_1} \frac{m}{n^{1/5}} -
    \frac{3}{\sqrt{c_1}} -6\frac{c_2}{(c_1)^{3/2} n^{-1/5}} \right) =
  \degree(v)\cdot \left( m   + 2 \frac{c_2}{c_1} \frac{m}{n^{1/5}} - O(1)  \right). 
\end{equation*}
Therefore, on average, vertex $v \in T$ of degree $\degree(v)$ receives more than
the threshold $\thresholdCount_v = m \cdot \degree(v) + 2m \cdot \degree(v)/n^{1/4}$ number
of walks for large enough $n$. Thus, some vertex in~$T$ will receive more
than~$\thresholdCount_v$ walks and 
output \emph{Reject}.

\newcommand{\notSampled}{\mathcal{E}} Let $\notSampled$ be the event that none
of the vertices in $P$ is sampled to be one of the starting points in $Q$.
Since each vertex $u \in V$ is sampled with probability $5000 \cdot\degree(v)/(2\epsilon \cdot m)$
and~$\volume(P) \ge (5\epsilon \cdot m)/4860$, it follows that
\begin{equation*}
    Pr[\notSampled] \le \left(1 - 5000 \cdot\degree(v)/(2\epsilon \cdot m) \right)^{\frac{5\epsilon \cdot m}{4860}}
    \le e^{-2.5}  = 0.08.
\end{equation*}
Taking a union bound over the probability of the event $\notSampled$ and the
probability of set $T$ around a starting vertex $u \in P$ not receiving enough
walks, we get that with probability at most~$0.1 + 0.08 =0.18$, no vertex will output \emph{Reject}.  Thus, our distributed algorithm will output \emph{Reject} with probability at least $2/3$, on at least one vertex of $G$.
 Finally, the upper bound on the number of communication rounds follows from the
length~$\ell = \frac{32}{\alpha^2} \log n$ of each random walk.
\end{proof}

%% file: conclusions-arxiv2024.tex
 \section{Conclusion and Future Directions}
 \label{sec:conclusion}
 This paper proposes a very simple distributed CONGEST model algorithm to test the conductance of a network.  At the end of the algorithm some node outputs \emph{Reject} if  the network is a bad conductor, or every node outputs \emph{Accept} otherwise. Unlike previous work, this shows we can achieve such testing without aggregation and spanning structures. This raises the possibility of conducting the testing even in the presence of node and link failures, in particular, if we can prove robustness of random walks (possibly in special cases). Since the algorithm takes $O(\log n)$ rounds, $O(m \log n)$ is an obvious bound on the number of messages - however, could a more careful analysis of the message complexity of this algorithm be useful? Property testing has the limit that there's a regime of uncertainity where the tester has no guarantees - a stronger algorithm would be a verifier which could verify the property over the whole range of input. Could verifiers be developed with similar complexities?



%% file: Conductance-Arxiv2024.bbl
\newcommand{\etalchar}[1]{$^{#1}$}
\begin{thebibliography}{DSHK{\etalchar{+}}11}

\bibitem[Alo86]{Alon86eigenvalues}
Noga Alon.
\newblock Eigenvalues and expanders, 1986.
\newblock URL: \url{https://doi.org/10.1007/BF02579166}.

\bibitem[AM84]{AlonMilman}
Noga Alon and V.~D. Milman.
\newblock Eigenvalues, expanders and superconcentrators (extended abstract).
\newblock In {\em 25th Annual Symposium on Foundations of Computer Science}.
  {IEEE} Computer Society, 1984.
\newblock URL: \url{https://doi.org/10.1109/SFCS.1984.715931}.

\bibitem[BPS09]{brakerskiPatt-Shamir}
Zvika Brakerski and Boaz Patt-Shamir.
\newblock Distributed discovery of large near-cliques.
\newblock In {\em Distributed Computing}, 2009.

\bibitem[CHFSV19]{Censor-Hillel2019}
Keren Censor-Hillel, Eldar Fischer, Gregory Schwartzman, and Yadu Vasudev.
\newblock {Fast Distributed Algorithms for Testing Graph Properties}.
\newblock {\em Distributed Computing}, 2019.

\bibitem[Chu97]{chung97}
Fan R.~K. Chung.
\newblock {\em Spectral Graph Theory}.
\newblock AMS, Providence, RI, 1997.

\bibitem[CS10]{CzumajS10}
Artur Czumaj and Christian Sohler.
\newblock {Testing Expansion in Bounded-Degree Graphs}.
\newblock {\em Combinatorics, Probability {\&} Computing}, 2010.

\bibitem[Dod84]{Dodziuk}
Jozef Dodziuk.
\newblock Difference equations, isoperimetric inequality and transience of
  certain random walks.
\newblock {\em Transactions of the American Mathematical Society}, 284(2),
  1984.
\newblock URL: \url{http://www.jstor.org/stable/1999107}.

\bibitem[DSHK{\etalchar{+}}11]{hardnessOfVerification}
Atish Das~Sarma, Stephan Holzer, Liah Kor, Amos Korman, Danupon Nanongkai,
  Gopal Pandurangan, David Peleg, and Roger Wattenhofer.
\newblock Distributed verification and hardness of distributed approximation.
\newblock In {\em Proceedings of the Forty-Third Annual ACM Symposium on Theory
  of Computing}, STOC '11, 2011.
\newblock URL: \url{https://doi.org/10.1145/1993636.1993686}.

\bibitem[EFF{\etalchar{+}}17]{even_et_al2017}
Guy Even, Orr Fischer, Pierre Fraigniaud, Tzlil Gonen, Reut Levi, Moti Medina,
  Pedro Montealegre, Dennis Olivetti, Rotem Oshman, Ivan Rapaport, and Ioan
  Todinca.
\newblock {Three Notes on Distributed Property Testing}.
\newblock In {\em 31st International Symposium on Distributed Computing
  (DISC)}, 2017.
\newblock URL: \url{http://drops.dagstuhl.de/opus/volltexte/2017/7984}.

\bibitem[FGO17]{FischerGO17}
Orr Fischer, Tzlil Gonen, and Rotem Oshman.
\newblock {Distributed Property Testing for Subgraph-Freeness Revisited}.
\newblock {\em CoRR}, 2017.
\newblock URL: \url{http://arxiv.org/abs/1705.04033}.

\bibitem[FV18]{VasudevDistributed}
Hendrik Fichtenberger and Yadu Vasudev.
\newblock A two-sided error distributed property tester for conductance.
\newblock In {\em 43rd International Symposium on Mathematical Foundations of
  Computer Science, {MFCS} 2018}, 2018.
\newblock URL: \url{https://doi.org/10.4230/LIPIcs.MFCS.2018.19}.

\bibitem[GGR96]{propertyTesting1}
O.~Goldreich, S.~Goldwasser, and D.~Ron.
\newblock Property testing and its connection to learning and approximation.
\newblock In {\em Proceedings of 37th Conference on Foundations of Computer
  Science}, 1996.
\newblock \href {http://dx.doi.org/10.1109/SFCS.1996.548493}
  {\path{doi:10.1109/SFCS.1996.548493}}.

\bibitem[Gol10]{graphPropertyTesting}
Oded Goldreich.
\newblock {\em Introduction to Testing Graph Properties}.
\newblock Springer Berlin Heidelberg, 2010.
\newblock URL: \url{https://doi.org/10.1007/978-3-642-16367-8_7}.

\bibitem[Gol17]{propertyTesting3}
Oded Goldreich.
\newblock {\em Introduction to Property Testing}.
\newblock Cambridge University Press, 2017.
\newblock \href {http://dx.doi.org/10.1017/9781108135252}
  {\path{doi:10.1017/9781108135252}}.

\bibitem[GR97]{propertyTesting2}
Oded Goldreich and Dana Ron.
\newblock Property testing in bounded degree graphs.
\newblock In {\em Proceedings of the Twenty-Ninth Annual ACM Symposium on
  Theory of Computing}, STOC '97, 1997.
\newblock URL: \url{https://doi.org/10.1145/258533.258627}.

\bibitem[GR00]{GoldreichRonTester}
Oded Goldreich and Dana Ron.
\newblock On testing expansion in bounded-degree graphs.
\newblock {\em Electron. Colloquium Comput. Complex.}, 2000.
\newblock URL:
  \url{https://eccc.weizmann.ac.il/eccc-reports/2000/TR00-020/index.html}.

\bibitem[GT12]{GharanT12}
Shayan~Oveis Gharan and Luca Trevisan.
\newblock Approximating the expansion profile and almost optimal local graph
  clustering.
\newblock In {\em 53rd Annual {IEEE} Symposium on Foundations of Computer
  Science, {FOCS} 2012}, 2012.
\newblock URL: \url{https://doi.org/10.1109/FOCS.2012.85}.

\bibitem[KPP{\etalchar{+}}13]{leaderElection}
Shay Kutten, Gopal Pandurangan, David Peleg, Peter Robinson, and Amitabh
  Trehan.
\newblock Sublinear bounds for randomized leader election.
\newblock In {\em Distributed Computing and Networking}, 2013.

\bibitem[KS11]{KaleS11}
Satyen Kale and C.~Seshadhri.
\newblock An expansion tester for bounded degree graphs.
\newblock {\em {SIAM} J. Comput.}, 40, 2011.
\newblock URL: \url{https://doi.org/10.1137/100802980}.

\bibitem[LMOS20]{MPCtester}
Jakub \L{}\k{a}cki, Slobodan Mitrovi\'{c}, Krzysztof Onak, and Piotr Sankowski.
\newblock Walking randomly, massively, and efficiently.
\newblock In {\em Proceedings of the 52nd Annual ACM SIGACT Symposium on Theory
  of Computing}, STOC 2020, 2020.
\newblock URL: \url{https://doi.org/10.1145/3357713.3384303}.

\bibitem[LP15]{LiP15}
Angsheng Li and Pan Peng.
\newblock Testing small set expansion in general graphs.
\newblock In {\em 32nd International Symposium on Theoretical Aspects of
  Computer Science, {STACS} 2015}, 2015.
\newblock URL: \url{https://doi.org/10.4230/LIPIcs.STACS.2015.622}.

\bibitem[LPP11]{Li2011TestingCI}
Angsheng Li, Yicheng Pan, and Pan Peng.
\newblock Testing conductance in general gr phs.
\newblock {\em Electron. Colloquium Comput. Complex.}, TR11, 2011.
\newblock URL: \url{https://api.semanticscholar.org/CorpusID:15721089}.

\bibitem[MP17]{MP-Mixing-ICDCN17}
Anisur~Rahaman Molla and Gopal Pandurangan.
\newblock Distributed computation of mixing time.
\newblock In {\em Proceedings of the 18th International Conference on
  Distributed Computing and Networking}, ICDCN '17. Association for Computing
  Machinery, 2017.
\newblock URL: \url{https://doi.org/10.1145/3007748.3007784}.

\bibitem[NS10]{AsafAsaf10}
Asaf Nachmias and Asaf Shapira.
\newblock Testing the expansion of a graph.
\newblock {\em Information and Computation}, 2010.
\newblock URL: \url{https://doi.org/10.1016/j.ic.2009.09.002}.

\bibitem[Sin93]{Sinclair}
Alistair Sinclair.
\newblock {\em Algorithms for Random Generation and Counting: A Markov Chain
  Approach}.
\newblock Birkhauser Verlag, 1993.

\bibitem[SMP15]{MP-SparseCut}
Atish~Das Sarma, Anisur~Rahaman Molla, and Gopal Pandurangan.
\newblock Distributed computation of sparse cuts via random walks.
\newblock In {\em Proceedings of the 16th International Conference on
  Distributed Computing and Networking}, ICDCN '15. Association for Computing
  Machinery, 2015.
\newblock URL: \url{https://doi.org/10.1145/2684464.2684474}.

\end{thebibliography}
